\documentclass[fleqn]{amsart}
\usepackage{amssymb,amsmath,amsthm,hyperref,setspace}

\newcommand{\R}{\mathbf{R}}

\renewcommand{\AA}{\mathcal{A}}
\newcommand{\BB}{\mathcal{B}}
\newcommand{\CC}{\mathcal{C}}

\newcommand{\HH}{\mathcal{H}}
\newcommand{\PP}{\mathcal{P}}

\newcommand{\ran}{\mathop{\mathrm{ran}}}
\renewcommand{\ker}{\mathop{\mathrm{ker}}}
\newcommand{\Var}{\mathop{\rm{Var}}}
\newcommand{\Cov}{\mathop{\rm{Cov}}}
\newcommand{\corr}{\mathop{\rm{corr}}}

\providecommand{\norm}[1]{\lVert#1\rVert}
\providecommand{\pair}[1]{\langle#1\rangle}

\newcommand{\onehalf}{\frac{1}{2}}

\newtheorem{definition}{Definition}[section]
\newtheorem{theorem}{Theorem}[section]
\newtheorem{corollary}[theorem]{Corollary}
\newtheorem{lemma}[theorem]{Lemma}
\newtheorem{example}{Example}

\title{A Simple Proof of the Fundamental Theorem of Asset Pricing}
\author{Keith A. Lewis}
\date{\today}
\begin{document}
\begin{abstract}
A simple statement and accessible proof of a version of the Fundamental
Theorem of Asset Pricing in discrete time is provided. Careful distinction
is made between prices and cash flows in order to provide uniform treatment
of all instruments. There is no need for a ``real-world''
measure in order to specify a model for derivative securities,
one simply specifies an arbitrage free model, tunes it to market data,
and gets down to the business of pricing, hedging, and managing the
risk of derivative securities.
\end{abstract}
\address{KALX, LLC \tt{\url{http://kalx.net}}}
\thanks{Peter Carr is entirely responsible for many
enjoyable and instructive discussions on this topic.
Andrew Kalotay provided background on Edward Thorpe and his contributions.
Alex Mayus provided practitioner insights.
Robert Merton graciously straightened me out on the early history.
Walter Schachermeyer provided background on the
technical aspects of the state of the art proofs.
I am entirely responsible for any omissions and errors.
}

\maketitle

\section{Introduction}
It is difficult to write a paper about the Fundamental Theorem of Asset
Pricing that is longer than the bibliography required to do justice to
the excellent work that has been done elucidating the key insight Fischer
Black, Myron Scholes, and Robert Merton had in the early '70's. At that
time, the Capital Asset Pricing Model and equilibrium reasoning dominated
the theory of security valuation so the notion that the relatively weak
assumption of no arbitrage could have such detailed implications about
possible prices resulted in well deserved Nobel prizes.

One aspect of the development of the FTAP has been the
technical difficulties involved in providing rigorous proofs and the
the increasingly convoluted statements of the theorem.
The primary contribution of this paper is a statement of the fundamental
theorem of asset pricing that is comprehensible to traders and risk
managers and a proof that is accessible to students at graduate level
courses in derivative securities. Emphasis is placed on distinguishing
between prices and cash flows in order to give a unified treatment of
all instruments. No artificial ``real world'' measures which are then
changed to risk-neutral measures needed. (See also Biagini and Cont
\cite{BiaCon2006}.)  One simply finds appropriate price deflators.

Section 2 gives a brief review of the history of the FTAP with an
eye to demonstrating the increasingly esoteric mathematics involved.
Section 3 states and proves the one period version and introduces a
definition of arbitrage more closely suited to what practitioners would
recognize. Several examples are presented to illustrate the usefulness
of the theorem. In section 4 the general result for discrete time models
is presented together with more examples.  The last section finishes
with some general remarks and a summary of the methodology proposed
in this paper. The appendix is an attempt to clairify attribution of
early results.

\section{Review}
From Merton's 1973 \cite{Mer1973} paper,
``The manifest characteristic of (21) is the number of variables
that it does {\em not} depend on'' where (21)
refers to the Black-Scholes 1973 \cite{BlaSch1973} option pricing formula
for a call having strike $E$ and expiration $\tau$
\begin{equation*}
f(S, \tau; E) = S\Phi(d_1) - Ee^{-rt}\Phi(d_1 - \sigma\sqrt{\tau}).
\end{equation*}
Here, $\Phi$ is the cumulative standard normal distribution, $\sigma^2$ is the
instantaneous variance of the return on the stock and $d_1 = [\log(S/E)
+ (r + \onehalf\sigma^2)\tau]/\sigma\sqrt{\tau}$.
In particular, the return on the stock does
not make a showing, unlike in the Capital Asset Pricing Model
where it shares center stage with covariance.
This was the key insight in the connection between arbitrage-free
models and martingales.

In the section immediately following Merton's claim he calls into
question the rigor of Black and Scholes' proof and provides his own.
His proof requires the bond process to have nonzero quadratic variation.
Merton 1974 \cite{Mer1974} provides what is now considered to be the
standard derivation.

A special case of the valuation formula that European option prices
are the discounted expected value of the option payoff under the risk
neutral measure makes its first appearance in the Cox and Ross 1976
\cite{CoxRos1976} paper.  The first version of the FTAP in a form we would
recognize today occurs in a Ross 1978 \cite{Ros1978} where it is called
the Basic Valuation Theorem.  The use of the Hahn-Banach theorem in the
proof also makes its first appearance here, although it is not clear
precisely what topological vector space is under consideration.  The
statement of the result is also couched in terms of market equilibrium,
but that is not used in the proof. Only the lack of arbitrage in the
model is required.

Harrison and Kreps \cite{HarKre1979} provide the first rigorous
proof of the one period FTAP (Theorem 1) in a Hilbert space setting. They
are also the first to prove results for general diffusion processes with
continuous, nonsingular coefficients and make the premonitory statement
``Theorem 3 can easily be extended to this larger class of processes,
but one then needs quite a lot of measure theoretic notation to make a
rigorous statement of the result.''

The 1981 paper of Harrison and Pliska \cite{HarPli1981} is primarily
concerned with models in which markets are complete (Question 1.16),
however they make the key observation, ``Thus the parts of probability
theory most relevant to the general question (1.16) are those results,
usually abstract in appearance and French in origin, which are invariant
under substitution of an equivalent measure.'' This observation applies
equally to incomplete market models and seems to have its genesis in
the much earlier work of Kemeny 1955 \cite{Kem1955} and Shimony 1955
\cite{Shi1955} as pointed out by W. Schachermeyer.

D. Kreps 1981 \cite{Kre1981} was the first to replace the assumption of
no arbitrage with that of no {\em free lunch}: ``The financial market
defined by $(X,\tau)$, $M$, and $\pi$ admits a free lunch if there are
nets $(m_\alpha)_{\alpha \in I} \in M_0$ and $(h_\alpha)_{\alpha\in I}
\in X_+$ such that $\lim_{\alpha\in I} (m_\alpha -  h_\alpha) = x$ for
some $x\in X_+\setminus\{0\}$.'' It is safe to say the set of traders
and risk managers that are able to comprehend this differs little from
the empty set.  It was a brilliant technical innovation in the theory but
the problem with first assuming a measure for the paths instrument prices
follow was that it made it difficult to apply the Hahn-Banach theorem. The
dual of $L^\infty(\tau)$ under the norm topology is intractable. The
dual of $L^\infty(\tau)$ under the weak-star topology is $L^1(\tau)$,
which by the Radon-Nikodym theorem can be identified with the set of
measures that are absolutely continuous with respect to $\tau$. This
is what one wants when hunting for equivalent martingale measures,
however one obstruction to the proof is that the positive functions
in $L^\infty(\tau)$ do not form a weak-star open set. Krep's highly
technical free lunch definition allowed him to use the full plate of
open sets available in the norm topology that is required for a rigorous
application of the Hahn-Banach theorem.

The escalation of technical machinery continues in Dalang, Morton
and Willinger 1990 \cite{DalMorWil1990}. This paper gives a rigorous
proof of the FTAP in discrete time for an arbitrary probability space
and is closest to this paper in subject matter.  They correctly point
out an integrability condition on the price process is not economically
meaningful since it is not invariant under change of measure.  They give
a proof that does not assume such a condition by invoking a nontrivial
measurable selection theorem. They also mention, ``However, if in addition
the process were assumed to be bounded, ...'' and point out how this
assumption could simplify their proof. The robust arbitrage definition
and the assumption of bounded prices is also used the original paper,
Long Jr. 1990 \cite{Lon1990},  on numeraire portfolios.

The pinnacle of abstraction comes in Delbaen and Schachermeyer 1994
\cite{DelSch1994} where they state and prove the FTAP in the continuous
time case. Theorem 1.1 states an equivalent martingale measure exists
if and only if there is {\em no free lunch with vanishing risk}: ``There
should be no sequence of final payoffs of admissible integrands, $f_n =
(H^n\cdot S)_\infty$, such that the negative parts $f_n^-$ tend to zero
uniformly and such that $f_n$ tends almost surely to a $[0,\infty]$-valued
function $f_0$ satisfying $P[f_0 > 0] > 0$.''  The authors were completely
correct when they claim ``The proof of Theorem 1.1 is quite technical...''

The fixation on change of measure and market completeness resulted in
increasingly technical definitions and proofs.  This paper presents a new
version of the Fundamental Theorem of Asset Pricing in discrete time.
No artificial probability measures are introduced and no ``change of
measure'' is involved.  The model allows for negative prices and for
cash flows (e.g., dividends, coupons, carry, etc.) to be associated with
instruments. All instruments are treated on an equal basis and there is
no need to assume the existence of a risk-free asset that can be used
to fund trading strategies.

As is customary, perfect liquidity is assumed: every instrument can
be instantaneously bought or sold in any quantity at the given price.
What is not customary is that prices are bounded and there is no a
priori measure on the space of possible outcomes. The algebras of sets
that represent available information determine the price dynamics that
are possible in an arbitrage-free model.

\section{The One Period Model}
The one period model is described by a vector, $x\in\R^m$, representing
the prices of $m$ instruments at the beginning of the period, a set
$\Omega$ of all possible outcomes over the period, and a bounded function
$X\colon\Omega\to\R^m$, representing the prices of the $m$ instruments
at the end of the period depending on the outcome, $\omega\in\Omega$.

\begin{definition}
Arbitrage exists if there is a vector $\gamma\in\R^m$ such that
$\gamma\cdot x < 0$ and $\gamma\cdot X(\omega)\ge0$ for all $\omega\in\Omega$.
\end{definition}

The cost of setting up the position $\gamma$ is
$\gamma\cdot x = \gamma_1 x_1 + \cdots + \gamma_m x_m$.
This being negative means money is made
by putting on the position. When the position is liquidated at the end
of the period, the proceeds are $\gamma\cdot X$. This being non-negative
means no money is lost.

It is standard in the literature to introduce an arbitrary probability
measure on $\Omega$ and use the conditions $\gamma\cdot x = 0$ and 
$\gamma\cdot X\ge0$ with $E[\gamma\cdot X]>0$
to define an arbitrage opportunity, e.g.,
Shiryaev, Kabanov, Kramkov and Melnikov \cite{ShiKabKraMel1994}
section 7.3, definition 1.  Making nothing
when setting up a position and having a nonzero probability of making a
positive amount of money with no estimate of either the probability or
amount of money to be made is not a realistic definition of an arbitrage
opportunity.  Traders want to know how much money they make up-front
with no risk of loss after the trade is put on.  
This is what Garman \cite{Gar1985} calls {\em strong arbitrage}.

Define the \emph{realized return} for a position, $\gamma$, by
$R_\gamma = \gamma\cdot X/\gamma\cdot x$,
whenever $\gamma\cdot x\not=0$. If there exists $\zeta\in\R^m$
with $\zeta\cdot X(\omega) = 1$ for $\omega\in\Omega$ (a 
zero coupon bond) then the price is $\zeta\cdot x = 1/R_\zeta$.
Zero interest rates correspond to a realized return of 1.

Note that arbitrage is equivalent to
the condition $R_\gamma < 0$ on $\Omega$ for some $\gamma\in\R^m$.
In particular, negative interest rates do not necessarily imply arbitrage. 

The set of all arbitrages form a cone since this set is closed
under multiplication by a positive scalar and addition. The following 
version of the FTAP shows how to compute an arbitrage when it exists.

\begin{theorem}{(One Period Fundamental Theorem of Asset Pricing)}
Arbitrage exists if and only if $x$ does not belong to the smallest
closed cone containing the range of $X$. If $x^*$ is the nearest point in
the cone to $x$, then $\gamma = x^* - x$ is an arbitrage.
\end{theorem}

\begin{proof}
If $x$ belongs to the cone, it is arbitrarily close to a finite sum
$\sum_j X(\omega_j) \pi_j$, where $\omega_j\in\Omega$ and $\pi_j > 0$ for
all $j$. If $\gamma\cdot X(\omega) \ge0$ for all $\omega\in\Omega$ then $\gamma\cdot
\sum X(\omega_j) \pi_j \ge0$, hence $\gamma\cdot x$ cannot be negative. The
other direction is a consequence of the following with $\CC$ being the
smallest closed cone containing $X(\Omega)$.
\end{proof}

\begin{lemma}
If $\CC\subset\R^m$ is a closed cone and $x\not\in \CC$, then there
exists $\gamma\in\R^m$ such that $\gamma\cdot x < 0$ and $\gamma\cdot y \ge0$
for all $y\in \CC$.
\end{lemma}

\begin{proof}
This result is well known, but here is an elementary self-contained
proof.  Since $\CC$ is closed and convex, there exists $x^*\in \CC$
such that $\norm{x^* - x} \le \norm{y - x}$ for all $y\in \CC$.
We have $\norm{x^* - x} \le \norm{tx^* - x}$ for $t \ge 0$, so $0 \le
(t^2 - 1)\norm{x^*}^2 - 2(t - 1)x^{*}\cdot x = f(t)$. Because $f(t)$
is quadratic in $t$ and vanishes at $t = 1$, we have $0 = f'(1) =
2\norm{x^*}^2 - 2x^{*}\cdot x$, hence $\gamma\cdot x^* = 0$.  Now $0 <
\norm{\gamma}^2 = \gamma\cdot x^* - \gamma\cdot x$, so $\gamma\cdot x < 0$.

Since $\norm{x^* - x} \le \norm{ty + x^* - x}$ for $t \ge 0$ and $y\in
\CC$, we have $0 \le t^2\norm{y}^2 + 2ty\cdot(x^* - x)$. Dividing by $t$
and setting $t = 0$ shows $\gamma\cdot y \ge0$.  
\end{proof}

Let $B(\Omega)$ be the Banach algebra of bounded real-valued functions
on $\Omega$. Its dual, $B(\Omega)^* = ba(\Omega)$, is the space of
finitely additive measures on $\Omega$, e.g., Dunford and Schwartz \cite{DunSch1954}.
If $\PP$ is the set of non-negative measures in $ba(\Omega)$, then
$\{\pair{X,\Pi}:\Pi\in\PP\}$ is the smallest closed cone containing
the range of $X$, where the angle brackets indicate the dual pairing.
There is no arbitrage if and only if there exists a non-negative finitely
additive measure, $\Pi$, on $\Omega$ such that $x = \langle X,\Pi\rangle$.
We call such \(\Pi\) a {\em price deflator}.

If $V\in B(\Omega)$ is the payoff function of an instrument and $V =
\gamma\cdot X$ for some $\gamma\in\R^m$, then the cost of replicating the payoff
is $\gamma\cdot x = \langle \gamma\cdot X,\Pi\rangle = \langle V,\Pi\rangle$.
Of course the dimension of such perfectly replicating payoff functions
can be at most $m$. The second fundamental theorem of asset pricing
states that when there are complete markets, the price is unique. But
that never happens in the real world.

If a zero coupon bond, $\zeta\in\R^m$, exists then the riskless realized
return is $R = R_\zeta = 1/\Pi(\Omega)$. If we let $P = \Pi R$, then $P$
is a probability measure and $x = \langle X/R,P\rangle = EX/R$. With
\(V\) as in the previous paragraph, the cost of the replicating payoff 
is \(v = EV/R\), the expected discounted payoff.

\subsubsection{Managing Risk}
The current theoretical foundations of Risk Mangagment are lacking
\footnote{As empirically verified in September 2008}. The
classical theory assumes complete markets and perfect hedging
and fails to provide useful tools for quantitatively assessing how wishful
this thinking is. 

The main defect of most current risk measures is that they fail
to take into account active hedging. E.g., VaR\cite{Jor2006} assumes trades
will be held to some time horizon and only considers a percentile loss.
The only use to someone running a business that they might
lose \(X\) in \(n\) days with probability \(p\) if they do nothing
is to put a tick in a regulatory checkbox.

Multi-period models will be considered below, but a first step
is to measure the least squared error in the one-period model. 
Given any measure \(\Pi\) and any payoff \(V\in B(\Omega)\), we can minimize
\(\pair{(\gamma\cdot X - V)^2, \Pi}\). The solution is
\(\gamma = \pair{XX^T,\Pi}^{-1}\pair{XV,\Pi}\).
The least squared error is 
\[\min_\gamma \pair{(\gamma\cdot X - V)^2, \Pi} 
= \pair{V^2,\Pi} -  \pair{XV,\Pi}^T\pair{XX^T,\Pi}^{-1}\pair{XV,\Pi}.\]

In the case of a two instrument market \(X = (R,S)\) 
where \(R\) is the realized return on a zero coupon bond we get
\(\gamma = ((EV - n ES)/R, n)\) where \(n = \Cov(S,V)/\Var S\) 
and the expectation corresponds to the probability measure
\(P = \Pi R\). If we further assume \(x = (1, s)\) we have
\(\gamma\cdot x = EV/R - n (ES/R - s)\) and
the least squared error reduces to \(\sin^2\theta\Var(V)/R\) where
\(\cos\theta\) is the correlation of \(S\) with \(V\).

If \(\Pi\) is a price deflator we get the same answer
for the price as in the one-period model without the need
to involve the Hahn-Banach theorem.

\subsection{Examples}
This section illustrates consequences of the one period model.
Standard results follow from rational application of mathematics instead
of ad hoc arguments.

\begin{example}{(Put-Call parity)}
Let \(\Omega = [0,\infty)\), \(x = (1, s, c, p)\), and 
\(X(\omega) = (R, \omega, (\omega - k)^+, (k - \omega)^+)\). 
\end{example}

This models a bond with riskless realized return \(R\), a stock that
can take on any non-negative value, and a put and call with the same strike.
Take \(\gamma = (-k/R, 1, -1, 1)\). Since
\(\gamma\cdot X(\omega) = -k + \omega - (\omega - k)^+ + (k - \omega)^+ = 0\)
it follows \(0 = x\cdot \gamma = -k/R + s - c + p\) so
\(s - k/R = c - p\). 

This is the first thing traders check with any European option model.
Put-call parity does not hold in general for American options because
the optimal exercise time for each option is not necessarily the same.

\begin{example}{(Cost of Carry)}
Let $\Omega = [0,\infty)$, $x = (1, s, 0)$, 
and $X(\omega) = (R, \omega, \omega - f)$.
\end{example}
This models a bond with riskless realized
return $R$, a stock, and a forward contract on the stock with forward $f$.
The smallest cone containing the range of $X$ is spanned
by $X(0) = (R, 0, -f)$ and $\lim_{\omega\to\infty}X(\omega)/\omega = (0, 1, 1)$.
Solving $(1, s) = a(R, 0) + b(0, 1)$ gives $a = 1/R$ and $b = s$.
This implies $0 = -f/R + s$ so $f = Rs$.

\begin{example}{(Standard Binomial Model)}
Let $\Omega = \{d, u\}$, $0<d<u$, $x = (1, s, v)$
and $X(\omega) = (R, s\omega, V(s\omega))$, where $V$ is any given function.
\end{example}
This is the usual (MBA) parametrization for the one period binomial model with
a risk-less bond having realized return $R$, and a stock having price $s$
that can go to either $sd$ or $su$.  The smallest cone containing the
range of $X$ is spanned by $X(d)$ and $X(u)$.  Solving $(1, s) = aX(d)
+ bX(u)$ for $a$ and $b$ yields $a = (u - R)/R(u - d)$ and
$b = (R - d)/R(u - d)$.
The condition that $a$ and $b$ are non-negative implies $d
\le R \le u$. The no arbitrage condition on the third component implies
\begin{equation*}
  v = \frac{1}{R}\left(\frac{u - R}{u - d} V(sd)
    + \frac{R - d}{u - d} V(su)\right).
\end{equation*}
In a binomial model, the option is a linear combination of the bond
and stock. This is obviously a serious defect in the model.
Solving $V(sd) = mR + nsd$ and $V(su) = mR + nsu$ for $n$ we see
the number of shares of stock to purchase in order to replicate the option
is \(n = (V(su) - V(sd))/(su - sd)\). 
Note that if $V$ is a call spread consisting of long one call with strike
slightly greater than $sd$ and short one call with strike slightly less
than $su$, then $\partial v/\partial s = 0$ since $V'(sd) = 0 = V'(su)$.

\begin{example}{(Binomial Model)}
Let $\Omega = \{S^+, S^-\}$, $x = (1, s, v)$, and 
$X(\omega) = (R, \omega, V(\omega))$, where $V$ is any given function.
\end{example}
As above we find
\begin{equation*}
	v = \frac{1}{R}\left(\frac{S^+ - Rs}{S^+ - S^-} V(S^-) 
		+ \frac{Rs - S^-}{S^+ - S^-} V(S^+)\right)
\end{equation*}
and the number of shares of stock required to replicate the option is
$n = (V(S^+) - V(S^-)/(S^+ - S^-)$. Note $\partial v/\partial s = n$
indicates the number of stock shares to buy in order to replicate
the option.

\begin{example}
Let $\Omega = [90,110]$, $x = (1, 100, 6)$, 
and $X(\omega) = (1, \omega, \max\{\omega - 100, 0\})$.
\end{example}
This corresponds to
zero interest rate, a stock having price 100 that will certainly end with
a price in the range 90 to 110, and a call with strike 100. One might
think the call could have any price between 0 are 10 without entailing
arbitrage, but that is not the case.

This model is not arbitrage free.  The smallest cone containing the
range of $X$ is spanned by $X(90)$, $X(100)$, and $X(110)$. It is easy
to see that $x$ does not belong to this cone since it lies above the
plane determined by the origin, $X(90)$ and $X(110)$.

Using $e_b$, $e_s$, and $e_c$ as unit vectors in the bond, stock, and call
directions, $X(90) = e_b + 90e_s$ and $X(110) = e_b + 110e_s + 10e_c$.
Grassmann algebra\cite{Pea1999} yields $X(110)\wedge X(90)
= 90 e_b\wedge e_s + 110e_s\wedge e_b + 10e_c\wedge e_b + 900e_c\wedge e_s
= -900e_s\wedge e_c + 10e_c\wedge e_b - 20e_b\wedge e_s$. The vector
perpendicular to this is $-900e_b + 10e_s - 20e_c$.

After dividing by 10, we can read off an arbitrage from this: borrow 90
using the bond, buy one share of stock, and sell two calls. The amount
made by putting on this position is $-\gamma\cdot x = 90 - 100 + 12 = 2$. At
expiration the position will be liquidated to pays $\gamma\cdot X(\omega)
= -90 + \omega - 2\max\{\omega - 100, 0\} = 10 - |100 - \omega| \ge 0$
for $90\le\omega\le 110$.

\begin{example}
Let $\Omega = [90,110]$, $x = (100, 9.1)$, 
and $X(\omega) = (\omega, \max\{\omega - 100, \})$.
\end{example}

Eliminating the bond does not imply the call can have any price
between 0 and 10 without arbitrage. The position
$\gamma = (1, -11)$ is an arbitrage.

\begin{example}{(Normal Model)}
Let \(\Omega = (-\infty,\infty)\), \(x = (1, s)\), 
\(X = (R, S)\) with \(R\) scalar, and \(S\) normally distributed.
\end{example}
This model was developed by Louis Bachelier in his
1900 PhD Thesis\cite{Bac1900} with an implicit dependence on \(R\). 
Choose the parameterization \(S = Rs(1 + \sigma Z)\) where
where Z is standard normal and the price deflator is \(\Pi = P/R\)
where \(P\) is the probability measure underlying
\(Z\). This model is arbitrage free for any value of \(\sigma\),
however it does allow for negative stock values. As long as
\(\sigma\) is much smaller than \(s\) the probability of
negative prices is negligible. Every model has its limitations.

A useful formula is \(\Cov(N, f(M)) = \Cov(N,M)Ef'(M)\) whenever
\(M\) and \(N\) are jointly normal.  This follows from \(Ee^{\alpha N}
f(M) = Ee^{\alpha N} E f(M + \alpha\Cov(M,N))\), taking a derivative
with respect to \(\alpha\), then setting \(\alpha = 0\).

The price of a put option with strike \(k\) is 
\begin{align*}
p(k) &= E(k - S)^+/R\\
&= E(k - S)1(S \le k)/R\\
&= (k/R) P(S\le k) - (ES/R)1(S\le k)\\
&= (k/R - s)P(S\le k) + (\Var(S)/R) E\delta_k(S)
\end{align*}
since \(d1(k - s)^+/ds = -\delta_k(s)\), where \(\delta_k\) is a 
delta function with unit mass at \(k\).

Let \(\phi(z) = e^{-z^2/2}/\sqrt{2\pi}\) be
the standard normal density and \(\Phi(z) = \int_{-\infty}^z \phi(z)\,dz\)
be the cumulative standard normal distribution.
We have \(E\delta_k(S) = E\delta_k(Rs(1 + \sigma Z)) = \phi(z)/Rs\sigma\)
where \(z = (k/Rs - 1)/\sigma\) hence
\(p(k) = (k/R - s)\Phi(z) + s\sigma\phi(z)\).
For an at-the-money option, \(k = Rs\), this reduces to
\(p(k) = s\sigma/\sqrt{2\pi}\).

The hedge position in the underlying is 
\(\partial p(k)/\partial s = -ER1(S\le k)/R = -\Phi(z)\) so
the at-the-money hedge is to short \(1/2\) share of stock.

For a general European option with payoff \(p\) we have the delta
hedge is \(\Cov(S,f(S))/\Var(S) = Ep'(S)\). If \(p\) is linear
then we can find a perfect hedge so let's estimate the least
squared error for quadratic payoffs. Letting \(\mu_k
= E(S - f)^k\) be the \(k\)-th central moment, where \(f = Rs = ES\),
and using \(EZ^2 = 1\) and \(EZ^4 = 3\) we find

\begin{align*}
\Var(p(S)) &= \mu_2 p'(f)^2 + (\mu_4 - \mu_2^2)p''(f)^2/4\\
&= f^2\sigma^2 p'(f)^2 + f^4\sigma^4 p''(f)^2/2.\\
\end{align*}
Since \(\Cov(S,p(S)) = \Var(S)Ep'(S) = \Var(S)p'(f)\) we have
\begin{align*}
\corr(S,p(S)) &= 1/\sqrt{1 + f^2\sigma^2 p''(f)^2/2p'(f)^2}\\
&\approx 1 - f^2\sigma^2 p''(f)^2/4p'(f)^2\\
\end{align*}
if \(p'(f) > 0\)
so \(\sin\theta \approx f\sigma p''(f)/2p'(f)\) 
for small \(\sigma\). The least squared error is
\(\Var(p(S))\sin^2\theta/R \approx f^2\sigma^2p''(f)^2/4R\)
which is second order in \(\sigma\) and does not depend (strongly)
on \(p'(f)\).

If \(p'(f) = 0\) then the correlation is zero and the the best
hedge is a cash position equal to \(Ep(S)\). If \(p'(f) < 0\)
a similar estimate holds for the correlation tending to \(-1\).

A curious result is that the at-the-money correlation for a call
is constant:
\(\corr(S,(S - f)^+) = 1/\sqrt{2 - 2/\pi} \approx 0.856\)
independent of \(R\), \(s\), and \(\sigma\).
This follows from \(\Cov(S,p(S)) = \Var(S)/2\) and
\(\Var p(S) = \Var(S)(1/2 - 1/2\pi)\) where \(p(x) = (x - f)^+\).

One technique
traders use to smooth out gamma for at-the-money options is to extend 
the option expiration by a day or two. This gives a quantitative estimate
of how bad that hedge might be.

%
%

\subsection{An Alternate Proof}
The preceding proof of the fundamental theorem of asset pricing
does not generalized to multi-period models.

Define $A\colon\R^m\to \R\oplus B(\Omega)$ by $A\xi = -\gamma\cdot x \oplus
\gamma\cdot X$.  This linear operator represents the account statements
that would result from putting on the position $\gamma$ at the beginning
of the period and taking it off at the end of the period.
Define $\PP$
to be the set of $\{p \oplus P\}$ where $p > 0$ is in $\R$ and $P \ge 0$
is in $B(\Omega)$.  Arbitrage exists if and only if $\ran A = \{A\gamma
: \gamma\in \R^m\}$ meets $\PP$.  If the intersection is empty, then by
the Hahn-Banach theorem \cite{BanMaz1933} there exists a hyperplane
$\HH$ containing $\ran A$ that does not intersect $\PP$. Since we are
working with the norm topology, clearly $1\oplus 1$ is the center of an
open ball contained in $\PP$, so the theorem applies.  The hyperplane
consist of all $y\oplus Y\in \R\oplus B(\Omega)$ such that $0 = y\pi +
\langle Y,\Pi\rangle$ for some $\pi\oplus\Pi\in\R\oplus ba(\Omega)$.

First note that $\langle \PP, \pi\oplus\Pi\rangle$ cannot contain
both positive and negative values. If it did, the convexity of $\PP$
would imply there is a point at which the dual pairing is zero and thereby 
meets $\HH$. 
We may assume that the dual pairing is always positive and that $\pi
= 1$.  Since $0 = \pair{A\gamma, \pi\oplus\Pi} = \pair{ -\gamma\cdot
x, \pi} + \pair{\gamma\cdot X, \Pi}$ for all $\gamma\in\R^m$
it follows $x = \langle X,\Pi\rangle$ for the non-negative measure
$\Pi$. This completes the alternate proof.

This proof does not yield the arbitrage vector when it exists, however
it can be modified to do so. Define $\PP^+ = \{\pi\oplus\Pi : \langle
p\oplus P,\pi\oplus\Pi\rangle > 0, p\oplus P\in\PP\}$. The Hahn-Banach
theorem implies $\ran A\cap\PP \not=\emptyset$ if and only if $\ker
A^*\cap\PP^+ = \emptyset$, where $A^*$ is the adjoint of $A$ and
$\ker A^* = \{\pi\oplus\Pi : A^*(\pi\oplus\Pi) = 0\}$. If the later
holds we know $0 < \inf_{\Pi\ge0} \norm{-x + \langle X,\Pi\rangle}$
since $A^*(\pi\oplus\Pi) = -x\pi + \langle X,\Pi\rangle$. The same
technique as in the first proof can now be applied.




\section{Multi-period Model}
The multi-period model is specified by an increasing sequence of times
$(t_j)_{0\le j\le n}$ at which transactions can occur, a sequence
of algebras $(\AA_j)_{0\le j\le n}$ on the set of possible outcomes
$\Omega$ where $\AA_j$ represents the information available at time
$t_j$, a sequence of bounded $\R^m$ valued functions $(X_j)_{0\le
j\le n}$ with $X_j$ being $\AA_j$ measurable that represent the prices
of $m$ instruments, and a sequence of bounded $\R^m$ valued functions
$(C_j)_{1\le j\le n}$ with $C_j$ being $\AA_j$ measurable that represent
the cash flows associated with holding one share of each instrument over
the preceding time period. We further assume the cardinality of $\AA_0$
is finite, and the $\AA_j$ are increasing.

A {\em trading strategy} is sequence of bounded $\R^m$ valued functions
$(\Gamma_j)_{0\le j\le n}$ with $\Gamma_j$ being $\AA_j$ measurable
that represent the amount in each security purchased at time $t_j$.
Your {\em position} is $\Xi_j = \Gamma_0 + \cdots + \Gamma_j$, the
accumulation of trades over time. A trading strategy is
called {\em closed out} at time $t_j$ if $\Xi_j = 0$. Note in the one period case closed out trading strategies have the
form $\Gamma_0 = \gamma$, $\Gamma_1 = -\gamma$.

The amount your {\em account} makes at time \(t_j\) is
\(A_j = \Xi_{j-1}\cdot C_j - \Gamma_j\cdot X_j\), \(0\le j\le n\),
where we use the convention \(C_0 = 0\). The financial interpretation
is that at time \(t_j\) you receive cash flows based on the position
held from \(t_{j-1}\) to \(t_j\) and are charged for trading 
\(\Gamma_j\) shares at prices \(X_j\).



\begin{definition}
{\rm Arbitrage} exists if there is trading strategy
that makes a strictly positive amount on the initial trade
and non-negative amounts until it is closed out.
\end{definition}

We now develop the mathematical machinery required to state and prove
the Fundamental Theorem of Asset Pricing.

Let \(B(\Omega,\AA,\R^m)\) denote the Banach algebra of bounded 
\(\AA\) measurable functions on \(\Omega\) taking values in \(\R^m\). 
We write this as \(B(\Omega,\AA)\) when \(m = 1\).

Recall that if $B$ is a Banach algebra we can define the product 
\(yy^*\in B^*\) for \(y\in B\) and \(y^*\in B^*\) by 
\(\pair{x,yy^*} = \pair{xy,y^*}\) for \(x\in B\),
a fact we will use below. 

The standard statement of the FTAP
uses conditional expectation. This version uses restriction
of measures, a much simpler concept.
The conditional expectation of
a random variable is defined by \(Y = E[X|\AA]\) if and only
\(Y\) is \(\AA\) measurable and
\(\int_A Y\,dP = \int_A X\,dP\) for all \(A\in\AA\).
Using the dual pairing this says \(\pair{1_AY,P} =
\pair{1_AX,P}\) for all \(A\in\AA\).
Using the product just defined we can write this as
\(\pair{1_A,YP} = \pair{1_A,XP}\) so \(YP(A) = XP(A)\)
for all \(A\in\AA\). If \(P\) has domain \(\AA\) this
says \(YP = XP|_\AA\).

We need a slight generalization. If \(Y\) is
\(\AA\) measurable, \(P\) has domain \(\AA\), and
\(\pair{1_AY,P} = \pair{1_AX,Q}\) for all \(A\in\AA\),
then \(YP = XQ|_\AA\). There is no requirement that
\(P\) and \(Q\) be probability measures.

Let $\PP\subset\bigoplus_{j=0}^n B(\Omega, \AA_j)$ be the cone of all
$\oplus_j P_j$ such that $P_0 > 0$ and $P_j \ge 0$, $1\le j\le n$.
The dual cone, $\PP^+$ is defined to be the set of all 
$\oplus_j \Pi_j$ in
$\bigoplus_{j=0}^n ba(\Omega, \AA_j)$ such that
$\langle P,\Pi\rangle = \langle \oplus_j P_j, \oplus_j \Pi_j\rangle
= \sum_j\langle P_j,\Pi_j\rangle > 0$.

\begin{lemma}
The dual cone $\PP^+$ consists of $\oplus_j\Pi_j$ such that $\Pi_0 > 0$,
and $\Pi_j\ge 0$ for $1\le j\le n$.
\end{lemma}
\begin{proof}
Since $0 < \pair{P_0,\Pi_0}$ for $P_0 > 0$ we have $\Pi_0(A) > 0$
for every atom of $\AA_0$ so $\Pi_0 > 0$.  For every $\epsilon > 0$
and any $j > 0$ we have $0 < \epsilon \Pi_0(\Omega) + \pair{P_j,\Pi_j}$
for every $P_j \ge 0$.  This implies $\Pi_j \ge 0$. 
\end{proof}


\begin{theorem}{(Multi-period Fundamental Theorem of Asset Pricing)}
There is no arbitrage if and only if there exists $\oplus_i\Pi_i\in\PP^+$
such that
\[X_i \Pi_i = (C_{i + 1} + X_{i + 1})\Pi_{i + 1}|_{\AA_i},\quad 0\le i < n.\]
\end{theorem}

Note each side of the equation is a vector-valued measure and
recall \(\Pi|_\AA\) denotes the measure \(\Pi\) restricted to
the algebra \(\AA\).

\begin{proof}
Define \(A\colon\bigoplus_{i=0}^{n}B(\Omega, \AA_i, \R^m) 
\to \bigoplus_{i=0}^{n}B(\Omega, \AA_i)\) by
\(A = \bigoplus_{0\le i \le n} A_i\). Define \(\CC\) to
be the subspace of strategies that are closed out
by time \(t_n\).

With \(\PP\) as above, no arbitrage is equivalent to
\(A\CC\cap\PP=\emptyset\). Again, the norm topology ensures that $\PP$
has an interior point so the Hahn-Banach theorem implies there exists
a hyperplane $\HH = \{X\in\bigoplus_{i=0}^{n}B(\Omega, \AA_i): \pair{X,\Pi} = 0\}$ for some $\Pi = \oplus_0^n\Pi_i$
containing $A\CC$ that does not meet $\PP$. It is not possible that
$\langle\PP,\Pi\rangle$ takes on different signs. Otherwise the convexity
of $\PP$ would imply $0 = \langle P,\Pi\rangle$ for some $P\in\PP$
so we may assume $\Pi\in\PP^+$. Note
$0 = \pair{A(\oplus_i\Gamma_i),\oplus_i\Pi_i}
= \sum_{i=0}^n\pair{\Xi_{i - 1}\cdot C_i - \Gamma_i\cdot X_i,\Pi_i}$
for all $\oplus_i\Gamma_i\in\CC$. 
Taking closed out strategies of the
form $\Gamma_i = \Gamma$, $\Gamma_{i+1} = -\Gamma$ having all other
terms zero yields, where $\Gamma$ is $\AA_i$ measurable,
gives $0 = \pair{\Xi_{i - 1}\cdot C_i - \Gamma_i\cdot X_i,\Pi_i}
+ \pair{\Xi_{i}\cdot C_{i+1} - \Gamma_{i+1}\cdot X_{i+1},\Pi_{i+1}}
= \pair{- \Gamma\cdot X_i,\Pi_i}
+ \pair{\Gamma\cdot C_{i+1} + \Gamma\cdot X_{i+1},\Pi_{i+1}}
$,
hence $\pair{\Gamma, X_i \Pi_i} = \pair{\Gamma, (C_{i+1} + X_{i+1})\Pi_{i+1}}$
for all $\AA_i$ measurable $\Gamma$.
Taking \(\Gamma\) to be a characteristic function proves
\(X_i\Pi_i = (C_{i+1} + X_{i+1})\Pi_{i+1}|_{\AA_i}\) for \(0\le i < n\).
\end{proof}

A simple induction shows
\begin{corollary}
With notation as above,
\begin{equation}
X_j \Pi_j
  = \sum_{j < i < k} C_i\Pi_i|_{\AA_j} + (C_k + X_k)\Pi_k|_{\AA_j},\quad j < k.
\end{equation}
\end{corollary}

This corrects and generalizes formula (2) in chapter 2 of Duffie \cite{Duf1996}.
As we will see below, this corollary is the primary tool for constructing
arbitrage free models. In the case of zero cash flows and increasing algebras, 
the no arbitrage condition is equivalent to $(X_j\Pi_j)_{j\ge0}$
being a martingale, by a slight abuse of the word martingale.

A standard way to define models is to specify a measure \(P\) on
\(\Omega\) and price deflators of the form \(\Pi_i = D_iP\) for
some \(D_i\in B(\Omega,\AA_i)\). In this case we can write
\(X_i \Pi_i = (C_{i+1} + X_{i+1})\Pi_{i+1}|_{\AA_i}\)
as \(X_i D_i = E[(C_{i+1} + X_{i+1})D_{i+1}|\AA_i]\).

In the one period case there is no need to distinguish between price and
cash flows. In the multi-period case one can account for the cash flows,
as in Pliska \cite{Pli1997}, by stipulating the price decreases by the
amount of the cash flow. Explicitly distinguishing
between prices and cash flows provides a unified model that uniformly incorporates 
other cash flows such as bond coupons or foreign exchange carry.

We say a closed strategy, $\Gamma$, is {\em self-financing} if all but the first
and last component of $A\Gamma$ are zero. The cost at $t_0$ of creating the
cash flow $\Xi_{n-1}\cdot C_n - \Gamma_n\cdot X_n$ at $t_n$ is clearly $\Gamma_0\cdot X_0\Pi_0$.

\begin{lemma}
If \((\Gamma_j)\) is a closed out self-financing strategy then 
\[\pair{\Gamma_0\cdot X_0, \Pi_0}
= \pair{\Xi_{n-1}\cdot C_n - \Gamma_n\cdot X_n, \Pi_n}\].
\end{lemma}
\begin{proof}
First we show that \(\pair{\Gamma_0\cdot X_0, \Pi_0} 
= \pair{\Xi_j\cdot X_j, \Pi_j}\) 
for \(j < n\).
The result holds for \(j = 0\). Assume it holds for \(j\),
then using the FTAP and self-financing condition
\begin{align*}
\pair{\Xi_j\cdot X_j, \Pi_j} 
&= \pair{\Xi_j\cdot(C_{j+1} + X_{j+1}), \Pi_{j+1}}\\
&= \pair{\Gamma_{j+1}\cdot X_{j+1} + \Xi_j\cdot X_{j+1}, \Pi_{j+1}}\\
&= \pair{\Xi_{j+1}\cdot X_{j+1}, \Pi_{j+1}}.
\end{align*}
Finally, \(\pair{\Xi_{n-1}\cdot X_{n-1}, \Pi_{n-1}}
= \pair{\Xi_{n-1}\cdot (C_n + X_n), \Pi_n}
= \pair{\Xi_{n-1}\cdot C_n - \Gamma_n\cdot X_n, \Pi_n}\)
since \(\Xi_{n-1} = -\Gamma_n\) for closed strategies.
\end{proof}

This lemma shows that
if a European derivative has payoff $V\colon\Omega\to\R$ at $t_n$ and
we can find a closed self-financing portfolio $(\Gamma_j)_{0\le j < n}$
such that $\Xi_{n-1}\cdot C_n - \Gamma_n\cdot X_n = V$, then the cost of a 
the replicating strategy is $\pair{V, \Pi_n/\Pi_0}$.
Since \(\Gamma_0\cdot X_0 = \pair{V, \Pi_n/\Pi_0}\) we can compute the initial
hedge by taking the derivative with respect to market values
\(\Gamma_0 = (d/dX_0) \pair{V, \Pi_n/\Pi_0}\). 

This formula is the foundation of delta hedging derivative securities.
In general such a strategy does not exist, but we could use an
optimization criteron, e.g., best least squares fit, and use the
fitting error as a measure of hedging risk.

\subsection{Examples}

\begin{example}{(Short Rate Process)}
A {\em short rate (realized return) process} $(R_j)_{j\ge0}$ is a scalar valued adapted process that
defines instruments having price $X_j = 1$ and a single non zero
cash flow $C_{j+1} = R_j$ at time $t_{j+1}$. 
\end{example}
No arbitrage implies $\Pi_j = R_j\Pi_{j+1}|_{\AA_j}$ so $R_j =
\Pi_j/\Pi_{j+1}|_{\AA_j}$.  If the price deflators are predictable,
i.e., $\Pi_{j+1}$ is $\AA_j$ measurable, $j \ge 0$, then $R_j =
\Pi_j/\Pi_{j+1}$. In this case the short rate process determines the
price deflators $\Pi_j = \Pi_0/(R_0\cdots R_{j-1})$, $j>0$.

Assuming price deflators are predictable is a tame assumption.
It means that at any given time one can borrow or lend at a known
rate over the next period. Note these can be used to guarantee
self-financing strategies always exist.

This result is the foundation of fixed income derivatives.
The price of all other fixed income derivatives (with no
default) are constrained by the short rate process. 

\begin{example}{(Zero Coupon Bonds)}
A zero coupon bond has a single cash flow \(C_k = 1\) at maturity \(t_k\).
\end{example}
Since \(X_j\Pi_j = \Pi_k|_{\AA_j}\) for a bond maturing
at \(t_k\) we have its price at time \(t_j \le t_k\) is 
\(X_j \equiv D_j(k) = \Pi_k/\Pi_j|_{\AA_j} = \Pi_k|_{\AA_j}/\Pi_j\).
The price at and after maturity is 0. Note \(D_j(j+1) = 1/R_j\).
The function \(j\mapsto D_0(j)\) is called the
{\em discount} or {\em zero} curve.

\begin{example}{(Forward Rate Agreement)}
A forward rate agreement starting at \(t_i\)
has price \(X_i = 0\) and
two non-zero cash flows,
\(C_j = -1\) at \(t_j\) and \(C_k = 1 + F_i(j,k)\delta(j,k)\) at \(t_k\)
where \(\delta(j,k)\) is the {\em daycount fraction} for the
interval \([t_j,t_k]\).
\end{example}
The {\em day count basis} (Actual/360, 30/360, etc.) is a market convention
that determines the day count fraction and is approximately equal to the time
in years of the corresponding interval.

We have \(0 = -1\Pi_j|_{\AA_i} + (1 + F_i(j,k)\delta(j,k)\Pi_k|_{\AA_i}\) so 
\[
F_i(j,k) = \frac{1}{\delta(j,k)}\left(\frac{\Pi_j}{\Pi_k} - 1\right)|_{\AA_i}
= \frac{1}{\delta(j,k)}\left(\frac{D_i(j)}{D_i(k)} - 1\right).
\] 
Forward rates are determined
by zero coupon bond prices since they are a portfolio of such.

Note that if a zero coupon bond with maturity \(t_k\) is available
at time \(t_j\) then \(F_j(j,k) = (1/D_j(k) - 1)/\delta(j,k)\) is
the {\em forward rate} over the interval.

\begin{example}{(Bonds)}
A bond is specified by {\em calculation dates} \(t_0 < t_1 <
\cdots < t_n\), cash flows
\(C_j = c \delta_j\), \(0 < j < n\), and
\(C_n = 1 + c \delta_n\) where \(\delta_j = \delta(j-1,j)\).
\end{example}
The price at time \(t_0\) satisfies \(X_0\Pi_0 = c\sum_{j=1}^n
\delta_j\Pi_j|{\AA_0} + \Pi_n|{\AA_0}\) so 
\(X_0 = c\sum_j \delta_jD_0(j) + D_0(n)\).
A bond is {priced at par} if \(X_0 = 1\) in which case
\(c = (1 - D_0(n))/\sum_j \delta_jD_0(j)\) is the {\em par coupon}.

\begin{example}{(Swaps)}
A swap is specified by calculation dates \(t_0 < t_1 <
\cdots < t_n\) and cash flows
\(C_j = (c - F_{j-1}(j-1,j)) \delta_j\), \(0 < j \le n\)
\end{example}
There are many types of swaps. This one is more accurately described
as paying fixed and receiveing float without exchange of principal.
It is also common for the day count basis of the fixed and
floating legs to be different.

A fundamental fact about the floating cash flow stream is
\begin{align*}
\sum_{j = 1}^n F_{j-1}(j-1,j) \delta_j \Pi_j|_{\AA_0}
&= \sum_{j = 1}^n (\Pi_{j-1}/\Pi_j - 1)|_{\AA_j} \Pi_j|_{\AA_0}\\
&= \sum_{j = 1}^n (\Pi_{j-1} - \Pi_j)|_{\AA_j}|_{\AA_0}\\
&= \Pi_0 - \Pi_n|_{\AA_0}.
\end{align*}
This shows the value of the floating leg is the same as
receiving a cash flow of 1 at \(t_0\) and paying a cash flow
of 1 at \(t_n\). The intuition is that the initial cash flow
can be invested at the prevailing forward rate over each
interval and rolled
over while harvesting the floating payments until maturity.

Swaps are typically issued at \(t_0\) with price \(X_0 = 0\).
Using the above fact shows the swap par coupon is determined
by the same formula as for a bond. More generally, if
\(X_t = 0\) for \(t\le t_0\) and \(X_t = 0\) we write
\[
F^\delta_t(t_0,\dots,t_n) = \frac{D_t(t_0) - D_t(t_n)}
{\sum_{j=1}^n \delta(j-1,j)D_t(t_j)}
\]
for the par coupon at time \(t\) corresponding to the
underlying (forward starting) swap. Note we are using
the actual times instead of the index as arguments.
Also note that a one period swap is simply a forward rate agreement.

\begin{example}{(Futures)}
The price of a futures is always zero. Given an underlying index
\(S_k\) at expiration \(t_k\), they are quoted as having
`price' \(\Phi_j\) at \(t_j\) with the constraint
\(\Phi_k = S_k\) at \(t_k\).
Their cash flows are \(C_j = \Phi_j - \Phi_{j-1}\), \(j \le k\).
\end{example}
No arbitrage implies \(0 = (\Phi_{j+1} - \Phi_j)\Pi_{j+1}|_{\AA_j}\).
If the deflators are
predictable then \(\Phi_j = \Phi_{j+1}|_{\AA_j} = S_k|_{\AA_j}\). 
The standard way of making this statement is to say futures quotes are
a martingale. 

If we assume there is a probability measure \(P\)
on \(\Omega\) such that \(\Pi_t = D_tP\) for some \(D_t\)
, the {\em stochastic discount} to time \(t\),
that are bounded \(\AA_t\) measurable functions then we can
write \(\pair{X,\Pi_t} = EXD_t\).

If \(F\) is a forward and \(D\) is the stochastic discount to
expiration we have \(0 = E(F - f)D
= EF ED + \Cov(F,D) - f ED\) so the {\em convexity} is
\(\phi - f = -\Cov(F,D)/ED\), where \(\phi = EF\) is the
futures rate. In general \(F\) and \(D\) have negative
correlation so futures quotes are higher than forward rates.

In the equity world it is often assumed the price deflators
are not stochastic and \(\Pi_t = D(0,t) \equiv D(t)\) is given. 
The {\em (instantaneous) spot rate}, \(r(t)\), is defined
by \(D(t) = e^{-tr(t)}\) and the {\em (instantaneous)
forward rate}, \(f(t)\), by \(D(t) = e^{-\int_0^t f(s)\,ds}\).
We also write \(D_s(t) = D(t)/D(s)\) for the discount
from time \(s\) to \(t\).
Stock volatilities swamp
out any dainty assumptions of stochastic rates.

\begin{example}{(Generalized Ho-Lee Model\cite{HoLee1986})}
The short rate process is \(R_t = \phi(t) + \sigma(t)B_t\).
\end{example}
The original Ho-Lee model specifies a constant volatility. It
allows the discount curve to be fitted to market
data. As in the Bachelier model, it allows interest rates to
be negative, but it has a simple closed form solution
using the fact \(\exp(-\int_0^t \Theta(s)^2/2\,ds + \int_0^t \Theta(s)\,dB_s)\)
is a martingale plus \(d(\Sigma(t)B_t) = \Sigma(t)\,dB_t + \Sigma'(t) B_t\)
so \(\int_t^u \sigma(s)B_s\,ds 
= \Sigma(u)B_u - \Sigma(t)B_t - \int_t^u \sigma(s)\,dB_s\)
where \(\Sigma' = \sigma\)
\begin{align*}
E_t e^{-\int_t^u \sigma(s)B_s\,ds}
&= E_t e^{-(\Sigma(u)B_u - \Sigma(t)B_t) + \int_t^u \Sigma(s)\,dB_s }\\
&= e^{-(\Sigma(u)B_t - \Sigma(t)B_t)} E_t e^{-(\Sigma(u)B_u - \Sigma(u)B_t) + \int_t^u \Sigma(s)\,dB_s }\\
&= e^{-(\Sigma(u)B_t - \Sigma(t)B_t)} E_t e^{\int_t^u (\Sigma(s) - \Sigma(u))\,dB_s }\\
&= e^{-(\Sigma(u)B_t - \Sigma(t)B_t)} e^{\frac{1}{2}\int_t^u (\Sigma(s) - \Sigma(u))^2\,ds }\\
\end{align*}
and \(E_t\) denotes conditional expectation with respect to time \(t\).
The generalized Ho-Lee model has discount prices
\[
D_t(u) = e^{-\int_t^u \phi(s) - \frac{1}{2}(\Sigma(s) - \Sigma(u))^2\,ds + (\Sigma(u) - \Sigma(t))B_t}
\]
where we reparameterize by replacing \(\sigma(t)\) with \(-\sigma(t)\).
In case of constant volatility we have
\[
D_t(u) = e^{-\int_t^u \phi(s) - \frac{1}{2}\sigma^2(s - u)^2\,ds + \sigma(u - t)B_t}.
\]
This shows the convexity in the Ho-Lee model is \(\phi(t) - f(t) = \frac{1}{2}\sigma^2t^2$
which is quadratic in \(t\).

\begin{example}{(Forwards)}
A forward is a contract issued at time $s$ and maturing at time $t$
having price $X_s = 0$ and one nonzero cash flow
$C_t = S_t - F_s(t)$ at time $t$, where $S_t$ is the price
at $t$ of the underlying and $F_s(t)$ is the forward rate that
is specified at time $s$.
\end{example}
Assuming no dividends \(S_sD(s) = S_tD(t)|_{\AA_s}\) so
\(0 = (X_sD(s) = (S_t - F_s(t))D(t)|_{\AA_s} = S_sD(s) - F_s(t))D(t)\) and
we have \(F_s(t) = S_s/D_s(t)\). This is just the cost-of-carry
formula. In the presence of dividends \((d_j)\) at \((t_j)\)
this formula becomes 
\(F_s(t) = \sum_{s < t_j \le t} d_j|_{\AA_s}/D_s(t_j) + S_s/D_s(t)\).
Note dividends may be random.

Binomial models are based on a {\em random walk}. 
Let \(\Omega = \{\omega = (\omega_1,\dots,\omega_n)
: \omega_j\in\{0,1\}, 1\le j\le n\}\) and let \(P\) be the measure
on \(\Omega\) with \(P(\{\omega\}) = 1/2^n\) for \(\omega\in\Omega\).
The equivalence
relation \([\omega]_j = [\omega']_j\) if and only if
\(\omega_i = \omega'_i\) for \(i\le j\) gives a partition that
determines the atoms of the algebra \(\AA_j\).

Random walk is the discrete time stochastic process
\(W_j(\omega) = \omega_1 + \cdots + \omega_j\), \(1\le j\le n\).
Note \(P(W_j = k) = {n\choose k}/2^n\), \(EW_j = j/2\),
and \(\Var(W_j) = j/2 - j^2/2\). If we let
\(Z_j = 2W_j - j\) then \(EZ_j = 0\) and \(\Var(Z_j) = j\).

Define \([\omega]_{j0} = [\omega]_{j+1}\cap \{\omega_{j+1} = 0\}\)
and similarly for \([\omega]_{j1}\) so \([\omega]_j\) is the
disjoint union of \([\omega]_{j0}\) and \([\omega]_{j1}\).
It is easy to see \(Z_{j+1}P|_{\AA_j} = Z_jP\). More generally
\begin{align*}
f(Z_{j+1}(\omega))P([\omega]_j)
&= f(Z_{j+1}(\omega))P([\omega]_{j0}) + f(Z_{j+1}(\omega))P([\omega]_{j1})\\
&= f(Z_{j}(\omega) - 1)P([\omega]_j)/2 + f(Z_{j}(\omega) + 1)P([\omega]_j)/2\\
\end{align*}
so \(f(Z_{j+1})P|_{\AA_j} = \frac{1}{2}(f(Z_{j} - 1) + f(Z_{j} + 1))P\).

\begin{example}{(Multi-period Binomial Model)}
Fix the annualized realized return \(R > 0\), 
the initial stock price \(s\), the drift \(\mu\), and the
volatility \(\sigma\).
Define \(X_j = (R_j, S_j) =
(R^j, se^{\mu j + \sigma Z_j})\).
\end{example}
Many price deflators exist but we will look for one having the
form \(\Pi_j = R^{-j}P\). Clearly \(R^{j+1}\Pi_{j+1}|_{\AA_j} = R^j\Pi_j\).
Since \(S_{j+1}\Pi_{j+1}|_{\AA_j} 
= (e^\mu/R)\frac{1}{2}(e^{-\sigma} + e^{\sigma})S_j\Pi_j\),
the model is arbitrage free if 
\(e^\mu = R/\cosh\sigma\).

\begin{example}{(Geometric Brownian Motion)}
Fix the spot rate \(r\), the initial stock price \(s\)
the drift \(\mu\), and the volatility \(\sigma\). Let \(B_t\) be standard
Brownian motion and define \(X_t = (e^{rt}, se^{\mu t + \sigma B_t})$.
\end{example}
Let \(P\) be Brownian measure and recall \(M_t^\lambda = 
e^{-\lambda^2t/2 + \lambda B_t}\)
is a martingale.
Looking for deflators of the form \(e^{-rt}P\)
ensures \(e^{-rt}\Pi_t|_{\AA_s} = e^{-rs}\Pi_s\).
Since \(S_t\Pi_t = se^{(\mu - r)t + \sigma B_t}\),
the model is arbitrage free if \(\mu = r - \sigma^2/2\).

The forward value of a put option paying \(\max\{k - S, 0\}\) 
at the end of the period is 
\(E\max\{k - S, 0\} = kP(S\le k) - ES1(S\le k)
= kP(S\le k) - ESP(Se^{\sigma^2t}\le k)\)
where we use \(Ee^N f(N) = Ee^N Ef(N + \Var(N))\). 
(More generally, \(Ee^Nf(N_1,...) = Ee^N Ef(N + \Cov(N, N_1),...)\) 
if \(N\),\(N_1\), ... are jointly normal.) 
This can be written 
\(E\max\{k - S, 0\} = k P(Z \le z) - f P(Z \le z - \sigma t)\)
where \(z = \sigma t/2 + (1/\sigma)\log k/f\), and \(f = Rs\)
is the forward price of the stock.

For a European option with payoff \(p\) at time \(t\), the
value of the option is \(v = e^{-rt}Ep(S_t)\). The delta is
\begin{align*}
\partial v/\partial s &= e^{-rt}Ep'(S_t)e^{(r - \sigma^2/2)t + \sigma B_t}\\
&= Ep'(e^{\sigma^2t}S_t)\\
\end{align*}
and the gamma is
\begin{align*}
\partial^2 v/\partial s^2 &= Ep''(e^{\sigma^2t}S_t)e^{\sigma^2t}e^{(r - \sigma^2/2)t + \sigma B_t}\\
&= e^{(r + \sigma^2)t}Ep''(e^{2\sigma^2t}S_t).\\
\end{align*}
%
%

\subsection{Infinitely Divisible Distrbutions}
Brownian motion is characterized as a stochastic process having
increments that are independent, stationary, and normally distributed.
Dropping the last requirement characterizes L\'evy processes\cite{Ber1998}.
Knowing the distribution at time 1 determines the distribution at
all times and the distribution at any time is infinitely
divisible. 

Prior to L\'evy and Khintchine, Kolmogorov \cite{kol1932} derived a
parameterization for the characteristic function of infinitely divisible
distributions having finite variance. There exists a number \(\gamma\)
and a non-decreasing function \(G(x)\) such that

\[
\log Ee^{iuX} = i\gamma u + \int_{-\infty}^\infty K_u(x)\,dG(x),
\]
where \(K_u(x) = (e^{iux} - 1 - iux)/x^2\).
Note \(\phi'(u) = i\gamma + i\int_{-\infty}^\infty (e^{iux} - 1)/x\,dG(x)\)
and \(\phi''(u) = -\int_{-\infty}^\infty e^{iux}\,dG(x)\)
so \(E X = -i\phi'(0) = \gamma\) and \(\Var X = -\phi''(0)
= \int_{-\infty}^\infty dG(x) = G(\infty) - G(-\infty)\).

\begin{lemma}
If \(X\) is infinitely divisible with Kolmogorov parameters \(\gamma\)
and \(G\), then \(Ee^{isX}e^{iuX} = Ee^{isX}Ee^{iuX^*}\) where
\(X^*\) has Kolmogorov parameters \(\gamma^*
= \gamma + \int_{-\infty}^\infty (e^{isx} - 1)/x\,dG(x) = -i\phi'(s)\)
and \(dG^*(x) = e^{isx}dG(x)\).
\end{lemma}
\begin{proof}
We have
\begin{align*}
E e^{isX}e^{iuX} &= E e^{i\gamma(s+u) + \int_{-\infty}^\infty K_{s+u}(x)\,dG(x)}\\
	&= Ee^{isX} e^{i\gamma u
		+ \int_{-\infty}^\infty (K_{s+u}(x) - K_s(x))\,dG(x)}
\end{align*}
A simple calculation shows
\(K_{s+u}(x) - K_s(x) = iu(e^{isx} - 1)/x + e^{isx}K_u(x)\)
so \(E e^{isX}e^{iuX} = Ee^{isX}Ee^{iuX^*}\) where \(X^*\) is
infinitely divisible with Kolmogorov parameters
\(\gamma^* = -i\phi'(s)\) and \(dG^*(x) = e^{isx}\,dG(x)\).
\end{proof}

We call \(X^*\) the \(K\)-transform of \(X\).

If \(X\) is standard normal, then \(\gamma = 0\), \(G = 1_{[0,\infty)}\)
and \(\phi(u) = -u^2/2\) so \(\gamma^* = is\) and \(dG^* = dG\).
We have \(e^{-s^2/2}Ee^{iu(is + X)} = e^{-s^2/2}e^{-su - u^2/2}
= e^{-(s + u)^2/2} = Ee^{isX}e^{iuX}\).

\begin{corollary}
If \(f\) and its Fourier transform are integrable, then
\(Ee^{isX}f(X) = Ee^{isX} Ef(X^*)\) where \(X^*\) is the
K-transform of \(X\).
\end{corollary}
\begin{proof}
If \(f\) and its Fourier transform are integrable, then
\(f(x) = \int_{-\infty}^\infty e^{iux}\hat{f}(u)\,du/2\pi\), where
\(\hat{f}(u) = \int_{-\infty}^\infty e^{-iux}f(x)\,dx\) is the
Fourier transform of \(f\).
\begin{align*}
E e^{isX}f(X) &= \int_{-\infty}^\infty E e^{iuX}e^{iuX} \hat{f}(u)\,du/2\pi\\
	&= Ee^{isX}\int_{-\infty}^\infty E e^{iuX^*} \hat{f}(u)\,du/2\pi\\
	&= Ee^{isX} E f(X^*)
\end{align*}
\end{proof}

\begin{example}{(L\'evy Processes)}
Fix the spot rate \(r\), the initial stock price \(s\),
the drift \(\mu\), and the volatility \(\sigma\). Let \(L_t\) be a
L\'evy process and define \(X_t = (e^{rt}, se^{\mu t + \sigma L_t})$.
\end{example}
Again we look for deflators of the form \(e^{-rt}P\). If we define
the {\em cumulant} \(\kappa_t(s) = \log Ee^{sL_t}\) then \(\kappa_t(s)
= t\kappa_1(s)\) and \(e^{-t\kappa_1(\sigma) + \sigma L_t}\) is a martingale.
Since \(S_t\Pi_t = se^{(\mu - r)t + \sigma L_t}\),
the model is arbitrage free if \(\mu = r - \kappa_1(\sigma)\).

The formula for the forward value of put is
\(E(k - S_t)^+ = E(k - S_t)1(S_t \le k)
= k P(S_t\le k) - se^{rt}P(S_t^* \le k)\) where
\(S_t = se^{(r - \kappa_1(\sigma))t + \sigma L_t^*}\) and
\(L_t^*\) is the K-transform with \(is = \sigma\).

\section{Remarks}
\begin{itemize}
\item Not only do
traders want to know exactly how much they make upfront based on the
size of the position they put on, they and their
risk managers also want to hedge the subsequent gains they might make
under favorable market conditions.


\item Different counterparties have different short rate processes.
A large financial institution can fund trading strategies at a
more favorable rate than a day trader using a credit card.

\item As previously noted, $\partial v/\partial s \not= n$ in Example 2,
however $\partial (Rv)/\partial R = ns$ for both Example 2 and 3.
In words, the derivative of the future value of the option with
respect to realized return is the dollar delta.

\item It is not necessay to assume algebras for the prices and cash flows are are increasing.
If they are adapted to the algebras \((\BB_j)\) and
\(\BB_j\subseteq\AA_j\) for all \(j\) then \(X_j\Pi_j\) will
be well defined.
This is useful in order to model a recombining
tree. In the standard binomial model the atoms of $\BB_j$ are $\{W_j =
j - 2i\}$, $0\le i\le j$. This can be used to give a rigorous foundation
to path bundling algorithms, e.g., Tilley \cite{Til1993}.

\item This theory only allows bounded functions as models of prices and
positions. This corresponds to reality, but not to the classical
Black-Scholes/Merton theory. 
The fact that prices are bounded has no material consequences when
it comes to model implementation. An unbounded price
process can be replaced by one stopped at an arbitrarily large value. Since we can make the
probability of stopping vanishingly small, calculation of option prices
can be made arbitrarily close to those computed using the unbounded model.
Every model I have implemented had prices bounded by \(1.8 \times 10^{308}\).

\item Likewise, discrete time is not material problem since one could
model yocto second time steps. In fact, continuous time introduces
serious technical problems such as doubling strategies\cite{HarKre1979}.
Zeno wasn't the only one to distract
people's attention with this sort of casuistry. 

\item Measures being finitely additive is also not an issue. Countably
additive measures are also finitely additive and so all such
models fit into this framework. Interchanging limits and the Radon-Nikodym 
theorem for finitely additive measures are
more complicated than for countably additive measures, but these are not needed here. 

\item The examples show this theory has the same expressive power as the
standard theory and illustrates the usefulness of distinguishing prices
from cash flows to uniformly handle all types of instruments. 
There is no need to cook up a ``real world'' measure.
Not only does it ultimately get replaced, it adds technical
complications to the theory.

\end{itemize}

\section{Appendix: Origins}
While preparing this paper I had difficulty understanding who figured
out what when in the early theory. Cutting edge research is always messy.
This appendix is my attempt to clear that up and point out the
repercussions. Priority is the currency of academics, legacy
is the other side of that coin.

Currency is both sides of the coin for practitioners and I make my living
trying to provide them with tools they find useful. They usually
don't understand the subtleties of mathematical models but they 
know if the software implementation provides numbers that make sense.

As George Box said ``all models are wrong, but some are useful.''
Mathematical Finance is still in its infancy, but it has notched
up some significant victories. Dollar denominated fixed
income derivatives having maturity less than 4 years trade at
basis point spreads. Every bank has a different implementation,
but they all get the same answer. ``Practitioners'' in that market
can no longer rely on cunning and makeshift.

As Haug and Taleb\cite{HauTal2011} carefully delineate, the Black-Scholes
and even more sophisticated formulas were used well before 
Black, Scholes, and Merton showed up on the scene.
They underscore the importance of the no arbitrage condition
and are entirely correct that traders still
use ad hoc devices to produce numbers they find useful. 
Options are used to
determine model parameters and now play the role of primary securities 
in hedging more complex derivatives. Such is financial market progress.

However, they don't seem to appreciate the power of the mathematical
underpinnings. Ed Thorpe came up with a formula for calls and puts,
but didn't know how to extend that to price bonds with embedded
options\footnote{Andrew Kalotay, personal communication}. 
Academics have time to reflect on the paths blazed
by practicioners.
Exotic option pricing formulas require nontrivial
mathematics unobtainable through seat-of-the-pants methods.

It is beyond the scope of this appendix to review the tenor of the
time laid down by Markowitz\cite{Mar1952}, Tobin\cite{Tob1958}, Sharpe\cite{Sha1964}, Lintner\cite{Lin1965} and other pioneers 
in the field of quantitative finance, but they developed an
economic theory to quantify how diversification reduced risk. The
Capital Asset Pricing Model showed how
to create portfolios that could minimize systemic market risk.

Many of the fundamental results in the FTAP can be traced back
to Merton's unpublished, but widely circulated,
technical report\cite{Mer1970} that ultimately became chapter 11 in
his book on continuous time finance\cite{Mer1992}.
It uses a general equilibrium pricing model (intertemporal CAPM) 
to derive the Black-Scholes option model. His proof
did not require normally distributed returns or a quadratic
utility function, as CAPM did, foreshadowing Ross's Arbitrage
Pricing Theory\cite{Ros1976}.

Merton also derived what is now called the Black-Scholes partial
differential equation and showed
how individual sample paths could be used to model
prices directly instead of only considering expected values.
Black and Scholes introduced the idea of dynamic trading
when people were thinking in terms of portfolio selection.
They showed continuous time trading with
prices modeled by an It\=o diffusion allows perfect
replication and that the problem of estimating mean stock returns
was irrelevant to pricing options.

This had some deleterious knock on effects in the theory
of mathematical finance. Merton was so far ahead of his time
with the mathematical tools he introduced that generations
of people in his field overestimated the power of
mathematics when it came to modeling the complicated world
we live in. People that did not have his ability to understand
the math latched on to binomial models. Brownian motion is
a binomial model in wolves clothing. 

Haug and Taleb are on the right track when it comes to pointing
out the consequences of a theory that no practitioner would
find plausible. I embarrassed myself in my early career
when a trader asked me how to price a barrier option that was
triggered on the second touch. For some reason he
didn't buy my explanation about the infinite oscillatory behavior of
Brownian motion and that even using the 100th time it touched the barrier
would have the same theoretical price.

The work of Boyce and Kalotay \cite{BoyKal1979} was far
ahead of its time. They took a
practical Operations Research approach to modeling what happens
at the cash flow level, including counterparty credit and
tax considerations. Something clumsily being rediscovered in our
post September 2008 world.

The origin of the modern theory of derivative securities is based on
Stephen Ross's 1977 paper ``A Simple Approach to the Valuation of Risky Streams.'' He was the first to realize that the assumption of
no arbitrage and the Hahn-Banach theorem placed a constraint on the
dynamics of sample paths. It is a purely geometric result.
The price deflator is simply a positive measure used to find
a point in a cone. Normalizing that to a probability measure does 
not tell you the probability of anything, although the normalizing
factor does tell you the price of a zero coupon bond if your
model has one.

Ross's approach was not as rigorous as Merton's and
the attempts to place his results on sound mathematical
footing led to the the escalation of increasingly abstract mathematical
machinery outlined in the Review section. This paper
is an endeavor to provide a statement of the fundamental theorem
of asset pricing that practicioners can understand and a mathematically 
rigorous proof that is accessible to masters level students.


\bibliographystyle{plain}
\bibliography{ftapd}{}

\end{document}